\documentclass[11pt]{article}
\usepackage{amssymb,amsmath,amsthm,sectsty,url,ifpdf}
\usepackage[letterpaper,hmargin=0.972in,vmargin=1.0in]{geometry}

\ifpdf
\usepackage[bookmarks=true,pdfstartview=FitH,colorlinks,linkcolor=blue,filecolor=blue,citecolor=blue,urlcolor=blue]{hyperref}
\else
\usepackage[dvipdfm,bookmarks=true,pdfstartview=FitH,colorlinks,linkcolor=black,filecolor=black,citecolor=black,urlcolor=black]{hyperref}
\fi

\usepackage{cleveref,aliascnt}

\usepackage[boxed]{algorithm}
\usepackage[labelfont=bf,small]{caption}
\usepackage{multicol}
\theoremstyle{plain}
\newtheorem{theorem}{Theorem}[section]
\newtheorem{lemma}[theorem]{Lemma}

\crefname{fact}{Fact}{Facts}
\newtheorem{claim}[theorem]{Claim}
\crefname{claim}{Claim}{Claims}

\theoremstyle{definition}
\newtheorem{definition}[theorem]{Definition}

\makeatletter
\newtheorem*{rep@theorem}{\rep@title}
\newcommand{\newreptheorem}[2]{%
  \newenvironment{rep#1}[1]{%
    \def\rep@title{#2 \ref{##1}}%
    \begin{rep@theorem}}%
    {\end{rep@theorem}}}
\makeatother

\newreptheorem{theorem}{Theorem}

\sectionfont{\large}
\subsectionfont{\normalsize}
\numberwithin{equation}{section} 

\newcommand{\MyParagraph}[1]{\medskip \noindent {\bf #1}}
\renewcommand{\paragraph}{\MyParagraph}

\DeclareMathOperator*{\argmax}{argmax}

\newcommand{\ignore}[1]{}
\newcommand{\qedsymb}{\hfill{\rule{2mm}{2mm}}}
\renewenvironment{proof}{\begin{trivlist} \item[\hspace{\labelsep}{\bf
\noindent Proof.\/}] }{\qedsymb\end{trivlist}}
\newenvironment{MyEnumerate}[1]{\begin{enumerate}\setlength{\itemsep}{0.1cm}
\setlength{\parskip}{-0.05cm} #1}{\end{enumerate}}

\newenvironment{proofsketch}{\begin{trivlist} \item[\hspace{\labelsep}{\bf
\noindent Proof Sketch.\/}] }{\qedsymb\end{trivlist}}

\sectionfont{\large}
\subsectionfont{\normalsize}
\numberwithin{equation}{section} 
\newcommand{\aka} {also known as\ }
\renewcommand{\aka} {a.k.a.\ }
\newcommand{\wrt} {with respect to\ }
\newcommand{\resp}{resp.,\ }
\newcommand{\ie}  {i.e.,\ }

\newcommand{\abs}[1]{\left|#1\right|}

\newcommand\E{\mathop{\mathbb E}}

\newcommand{\eqdef} {\mathrel{\stackrel{\makebox[0pt]{\mbox{\normalfont\tiny
def}}}{=}}}

\newcommand{\U}{\mathbf U}
\newcommand{\N}{\mathbb N}
\newcommand{\R}{\mathbb R}

\newcommand{\Hc}{\mathcal H}

\newcommand{\B}{\{ 0,1 \}}

\def\epsilon{\varepsilon}
\def\eps{\varepsilon}

\newcommand{\D}{\mathcal D}

\newcommand{\Adv}{\mathcal A}

\newcommand{\ith}[1]{{#1}\textsuperscript{th}}

\newcommand{\negl}{{\mathsf{negl}}}
\newcommand{\poly}{{\mathsf{poly}}}
\newcommand{\polylog}{{\mathsf{polylog}}}

\newcommand{\SD}{\mathsf{SD}}

\def\({\left(}
\def\){\right)}

\newcommand{\GOOD}{\mathsf{Good}} 
\newcommand{\party}[1]{\mathsf{P}_{{#1}}}
\newcommand{\corr}{\textsf{corrupt}_k}

\newcommand{\row}{\textsf{Row}}

\newcommand{\lo}{{\mathsf{R}}}
\newcommand{\sh}{{\mathsf{r}}}

\newcommand{\out}[3]{{\mathsf{out}(#2_{#1} \mid #3)}}
\newcommand{\success}{{\mathsf{succ}}}
\newcommand{\successset}{{\mathsf{succ_{s}}}}

\newcommand{\valset}[2]{{\successset(#2_{#1})}}

\newcommand{\valb}[3]{{\success_{#1}(#3_{#2})}}

\newcommand{\supp}{\mathsf{supp}}

\newcommand{\ent}{\mathsf{entropy}}

\newcommand{\real}{\Adv}
\newcommand{\ideal}{\mathsf{ideal}}
\newcommand{\trans}{{\mathsf{Trans}}}
\newcommand{\Exp}{\mathsf{Exp}}
\newcommand{\MAP}{\mathsf{MAP}}
\newcommand{\event}{\mathsf{E}}

\theoremstyle{remark}

\def\({\left(}
\def\){\right)}

\title{Compressing Communication in Selection Protocols\footnote{A preliminary
    version of this work appeared in the 29th International Symposium on
    DIStributed Computing (DISC 2015), pp. 467--479.}}

\author{Yael Tauman Kalai \thanks{Microsoft Research. Email: {\tt
      yael@microsoft.com}.}
  \and
  Ilan Komargodski \thanks{Weizmann Institute of Science, Israel. Email: {\tt
      ilan.komargodski@weizmann.ac.il}. Part of this work done while an intern
    at MSR New England. Supported in part by a grant from the I-CORE Program of
    the Planning and Budgeting Committee, the Israel Science Foundation, BSF and
    the Israeli Ministry of Science and Technology.  } }

\date{}
\begin{document}

\maketitle

\begin{abstract}
  We show how to compress communication in selection protocols, where the goal is to agree on a sequence of random bits using only a broadcast channel. More specifically, we present a generic method for converting any selection protocol, into another selection protocol where each message is {\em short} while preserving the same number of rounds, the same output distribution, and the same resilience to error. Assuming that the output of the protocol lies in some universe of size~$M$, in our resulting protocol each message consists of only $\polylog(M,n,d)$ many bits, where $n$ is the number of parties and $d$ is the number of rounds.  Our transformation works in the presence of either static or adaptive Byzantine faults.

  As a corollary, we conclude that for any $\poly(n)$-round collective coin-flipping protocol, leader election protocol, or general selection protocols, messages of length $\polylog(n)$ suffice (in the presence of either static or adaptive Byzantine faults).

\end{abstract}


\thispagestyle{empty}
\newpage \setcounter{page}{1}

\section{Introduction}
The resource of communication is central in several fields of computer
science. We focus on minimizing this resource for selection protocols. A
selection protocol is a protocol over $n$ parties, each having a private source
of randomness, in which the goal of the parties is to agree on a sequence of
common random bits. We focus on the full information model \cite{BenOrL85},
where the parties communicate via a single broadcast channel. There is a global
counter which synchronizes parties in between rounds but they communicate
asynchronously withing rounds.  A selection protocol is a generalization of
several very well studied problems, including collective coin-flipping and
leader election.  

The challenge in designing such protocols is that a subset of the parties may be
corrupted and the rest of the parties should nevertheless agree on a random
output. We model faulty parties by a computationally unbounded adversary who
controls a subset of parties and whose aim is to bias the output of the
protocol. We assume that once a party is corrupted, the adversary gains complete
control over the party and can send any messages on its behalf, and the messages
can depend on the entire transcript so far.  In addition, we allow our adversary
to be \emph{rushing}, \ie it can schedule the delivery of the messages within
each round.  We consider two classes of adversaries: \emph{static} and
\emph{adaptive}. A static adversary is an adversary that chooses which parties
to corrupt ahead of time, before the protocol begins. An adaptive adversary, on
the other hand, is allowed to choose which parties to corrupt {\em adaptively}
in the course of the protocol as a function of the messages seen so far. We say
that a protocol is (statically/adaptively) \emph{secure} or \emph{resilient} if it results with a common
random output in the presence of a (statical/adaptive) adversary that corrupts parties.

We study the following question.
\begin{center}
  \textit{Is there a generic way to compress communication in selection
    protocols, without negatively affecting the round complexity,
    fault-tolerance and other resources?}
\end{center}
We give a positive answer to this question. Namely, we show how to compress
communication in selection protocols without incurring \emph{any} cost to the
round complexity or the resilience to errors. More details follow.

\paragraph{A concrete motivation: adaptively-secure coin-flipping.}
An important distributed task that was extensively studied in the full
information model, is that of \emph{collective coin-flipping}. In this problem,
a set of $n$ parties use private randomness and are required to generate a
common random bit. The goal of the parties is to jointly output a somewhat
uniform bit even in the case that some of the parties are faulty and controlled
by a static (resp.~adaptive) adversary whose goal is to bias the output of the protocol
in some direction.

This problem was first formulated and studied by Ben-Or and Linial
\cite{BenOrL85}. In the case of static adversaries, collective coin-flipping is
well studied and almost matching upper and lower bounds are known
\cite{Feige99,RussellSZ02}, whereas the case of adaptive adversaries has
received much less attention. Ben-Or and Linial \cite{BenOrL85} showed that the
majority protocol (in which each party sends a uniformly random bit and the
output of the protocol is the majority of the bits sent) is resilient to
$\Theta(\sqrt{n})$ adaptive corruptions. Furthermore, they conjectured that this
protocol is optimal, that is, they conjectured that any coin-flipping protocol
is resilient to at most $O(\sqrt{n})$ adaptive corruptions. Shortly afterwards,
Lichtenstein, Linial and Saks \cite{LichtensteinLS89} proved the conjecture for
protocols in which each party is allowed to send only \emph{one} bit. Very
recently, Goldwasser, Kalai and Park~\cite{GoldwasserKP15} proved a different
special-case of the aforementioned conjecture: any \emph{symmetric} (many-bit)
one-round collective coin-flipping protocol\footnote{A symmetric protocol $\Pi$
  is one that is oblivious to the order of its inputs: namely, for any
  permutation $\pi\colon[n]\to [n]$ of the parties, it holds that
  $\Pi(r_1,\dots,r_n) = \Pi(r_{\pi(1)},\dots, r_{\pi(n)} )$.} is resilient to at
most $\widetilde{O}(\sqrt{n})$ adaptive corruptions.  Despite all this effort,
proving a general lower bound, or constructing a collective coin-flipping
protocol that is resilient to at least $\omega(\sqrt{n})$ adaptive corruptions,
remains an intriguing open problem.

The result of \cite{LichtensteinLS89} suggests that when seeking for a
collective coin-flipping protocol that is resilient to at least
$\omega(\sqrt{n})$ adaptive corruptions, to focus on protocols that consist of
many communication rounds, or protocols in which parties send long messages.
Our main result (\Cref{lemma:many2few} below) is that long messages are not
needed in adaptively secure coin-flipping protocols with $\poly(n)$ rounds, and
messages of length $\polylog(n)$ suffice.\footnote{Note that if one could show
  that these $\polylog(n)$ bits can be sent bit by bit sequentially, then using
  the lower bound of \cite{LichtensteinLS89}, we could obtain that any
  collective coin flipping protocol in which each player sends $O(1)$ messages
  is resilient to at most $\sqrt{n\cdot \polylog(n)}$ adaptive
  corruptions. However, in the adaptive setting it is not clear that security is
  preserved if messages are sent bit by bit.} This is true more generally for
leader election protocols, and for selection protocols where the output comes
from a universe of size at most quasi-polynomial in~$n$.

\subsection{Our Results}
Our main result is that ``long'' messages are not needed for selection
protocols.  More specifically, we show how to convert any selection protocol,
whose output comes from a universe of size~$M$, into a selection protocol
with the same communication pattern\footnote{Here, we mean that a party sends a
  message at round $i$ of the new protocol only if it sends a message at round $i$ of
  the original protocol.}, the same output distribution, the same
security guarantees, and where parties send messages of length $\ell =
\polylog(M,n,d)$.  Note that for many well studied distributed tasks, such as
coin-flipping, leader election, and more, the output is from a universe of size
at most $\poly(n)$, in which case our result says that if we consider
$\poly(n)$-round protocols, then messages of length $\polylog(n)$ suffice.

\paragraph{Our results in more detail.}
Formally, we say that a selection protocol $\Pi$ is
\emph{$(t,\delta,s)$-statically (\resp adaptively) secure} if for any adversary
$\Adv$ that \emph{statically (\resp adaptively)} corrupts at most $t=t(n)$
parties, and any subset $S$ of the output universe such that $\abs{S}= s$, it
holds that
\begin{align*}
  \left|\Pr_{}\left[ \text{Output of $\Adv(\Pi)$} \in S \right] - \Pr_{}\left[
      \text{Output of $\Pi$} \in S \right]\right| \leq \delta,
\end{align*}
where ``Output of $\Adv(\Pi)$'' means the output of the protocol when executed
in the presence of the adversary $\Adv$, ``Output of $\Pi$'' means the output of
the protocol when executed honestly, and the probabilities are taken over the
internal randomness of the parties. In addition, we say that a protocol $\Pi$
\emph{simulates} a protocol $\Pi'$ if the outcomes of the protocols are
statistically close (when executed honestly) and their communication patterns
are the same.

Our main result is a generic communication compression theorem which, roughly
speaking, states that $(t,\delta,s)$-statically (\resp adaptively) secure
selection protocols \emph{do not} need ``long'' messages. Namely, we show that
any secure selection protocol which sends arbitrarily long messages can be
simulated by a protocol which is almost as secure and sends short messages. The
loss in security is a negligible (denoted by $\negl$), namely, asymptotically
smaller than any inverse polynomial function.

\begin{theorem}[Main theorem --- informal]\label{lemma:many2few}
  Any $(t,\delta,s)$-statically (\resp adaptively) secure selection protocol
  that outputs $m$ bits (or more generally, has an output universe of
  size~$2^m$), can be simulated by a $(t,\delta',s)$-statically (\resp
  adaptively) secure selection protocol, where $\delta' = \delta + \negl(n)$ and
  parties send messages of length $\ell = m\cdot \polylog(n,d)$.
\end{theorem}

We note that the transformation in Theorem~\ref{lemma:many2few} results in a
\emph{non-uniform} protocol, even if the protocol we started with is uniform.
We elaborate on this in Section~\ref{sec:techniques}.


\subsection{Overview of Our Techniques}\label{sec:techniques}

In this section we provide a high-level overview of our main ideas and
techniques. First, we observe that in our model of communication (the full
information model where all communication is done via a broadcast channel) one
can assume, without loss of generality, that any selection protocol (in which
parties do not have private inputs except a source of randomness), can be
transformed into a public-coin protocol, in which honest parties' messages
consist only of random bits. This fact is a folklore, and for the sake of
completeness we include a proof sketch of it in \Cref{sec:public_coin}.

Our main result is a generic transformation that converts any public-coin
protocol, in which parties send arbitrarily long messages, into a protocol in
which parties send messages of length $m\cdot \polylog(n\cdot d)$, where $m$ is the
number of bits the protocol outputs, $n$ is the number of parties participating
in the protocol, and $d$ is the number of communication rounds.  The resulting
protocol simulates the original protocol, has the same round complexity, and
satisfies the same security guarantees. Next, we elaborate on how this
transformation works.

Suppose for simplicity that in our underlying protocol each message sent is of
length~$L=L(n)$ (and thus the messages come from a universe of size~$2^L$), and
think of~$L$ as being very large.  We convert any such protocol into a new
protocol where each message consists of only~$\ell$ bits, where think of~$\ell$
as being significantly smaller than~$L$.  This is done by a priori choosing
$2^\ell$ messages within the $2^L$-size universe, and restricting the parties to
send messages from this restricted universe.  Thus, now each message is of
length $\ell$, which is supposedly significantly smaller than~$L$.  We note that
a similar approach was taken in \cite{Newman91} in the context of transforming
public randomness into private randomness in communication complexity, in
\cite{GoldreichS10} to reduce the number of random bits needed for property
testers, and most recently in \cite{GoldwasserKP15} to prove a lower bound for
coin-flipping protocols in the setting of strong adaptive adversaries.

A priori, it may seem that such an approach is doomed to fail, since by
restricting the honest parties to send messages from a small universe within the
large~$2^L$-size universe, we give the adversary a significant amount of
information about future messages (especially in the multi-round case).
Intuitively, the reason security is not compromised is that there are {\em many}
possible restrictions, and it suffices to prove that a few (or only one) of
these restrictions is secure. In other words, very loosely speaking, since we
believe that most of the bits sent by honest parties are not ``sensitive'', we
believe that it is safe to post some information about each message ahead of
time.

For the sake of simplicity, in this overview we focus on static adversaries, and
to simplify matters even further, we assume the adversary always corrupts the
first~$t$ parties. This simplified setting already captures the high-level
intuition behind our security proof in \Cref{sec:proof}.

Let us first consider one-round protocols. Note that for one-round protocols
restricting the message space of honest parties does not affect security at all
since we consider rushing adversaries, who may choose which messages to send
based on the content of the messages sent by all honest parties in that
round. Thus, reducing the length of messages is trivial in this case,
assuming the set of parties that the adversary corrupts is predetermined.  We
mention that even in this extremely simplified setting, we need $\ell$ to be
linear in~$m$ for correctness (``simulation''), \ie in order to ensure that the
output is distributed correctly.

Next, consider a multi-round protocol~$\Pi$.  We denote by~$H$ the restricted
message space, \ie $H$ is a subset of the message universe of size~$2^\ell$, and
denote by $\Pi_H$ the protocol $\Pi$, where the messages are restricted to the
set~$H$. Suppose that for any set~$H$ there exists an adversary $\Adv^H$ that
biases the outcome of $\Pi_H$, say towards~$0$.\footnote{Of course, it may be
  that for different sets~$H$, the adversary~$\Adv^H$ biases the outcome to a
  different value. For simplicity we assume here that all the adversaries bias
  the outcome towards a fixed message, which we denote by $0$.} We show that in
this case there exists an adversary~$\Adv$ in the underlying protocol that
biases the outcome towards~$0$.  Loosely speaking, at each step the
adversary~$\Adv$ will simulate one of the adversaries $\Adv^H$. More
specifically, at any point in the underlying protocol, the adversary will
randomly choose a set~$H$ such that the transcript so far is consistent (\ie
same transcript) with a run of protocol $\Pi_H$ with the adversary $\Adv^H$, and
will simulate the adversary $\Adv^H$.  The main difficulty is to show that with
high probability there exists such~$H$ (\ie the remaining set of consistent
$H$'s is non-empty).  This follows from a counting argument and basic
probability analysis.

In our actual construction, we have a distinct set $H$ of size~$2^\ell$
corresponding to {\em each} message of the protocol.  Thus, if the underlying
protocol~$\Pi$ has~$d$ rounds, and all the parties send a message in each round,
then the resulting (short-message) protocol is associated with $d\cdot n$ sets
$H_1,\ldots,H_{d\cdot n}$ each of size~$2^\ell$, where the message of the
$\ith{j}$ party in the $\ith{i}$ round is restricted to be in the set $H_{i,j}$.
We denote all these sets by a matrix $H\in\left(\{0,1\}^L\right)^{{d\cdot
    n}\times 2^\ell}$, where the row $(i,j)$ of $H$ corresponds to the set of
messages that the $\ith{j}$ party can send during the $\ith{i}$ round.

Note that there are $2^{L\cdot 2^\ell\cdot d\cdot n}$ such matrices.  Each time
an honest party sends a uniformly random message in~$\Pi$ it reduces the set of
consistent matrices by approximately a~$2^L$-factor (with high probability).
Any time the adversary $\Adv$ sends a message, it also reduces the set of
consistent matrices~$H$, since his message is consistent only with some of the
adversaries $\Adv^H$, but again a probabilistic argument can be used to claim
that it does not reduce the set of matrices by too much, and hence, with high
probability there always exist matrices~$H$ that are consistent with the
transcript so far.

We briefly mention that the analysis in the case of adaptive corruptions follows
the same outline presented above.  One complication is that the mere decision of
whether to corrupt or not reduces the set of consistent matrices~$H$. Nevertheless, we
argue that many consistent matrices remain.

We emphasize that the above is an over-simplification of our ideas, and the actual proof is more complex.
We refer to \Cref{sec:proof} for more details.

\section{Preliminaries}\label{sec:prelim}
In this section we present the notation and basic definitions that are used in
this work. For an integer $n \in \mathbb{N}$ we denote by $[n]$ the set
$\{1,\ldots, n\}$.  For a distribution $X$ we denote by $x \leftarrow X$ the
process of sampling a value $x$ from the distribution $X$. Similarly, for a set
$X$ we denote by $x \leftarrow X$ the process of sampling a
value $x$ from the uniform distribution over $X$.  Unless explicitly
stated, we assume that the underlying probability distribution in our equations
is the uniform distribution over the appropriate set. We let $\U_L$ denote the
uniform distribution over $\B^L$. We use $\log x$ to denote a logarithm in
base $2$.

A function
$\negl\colon\N\to\R$ is said to be \emph{negligible} if for every constant $c > 0$ there exists
an integer $N_c$ such that $\negl(n) < n^{-c}$ for all $n > N_c$.

The {\em statistical distance} between two random variables $X$ and $Y$ over a
finite domain $\Omega$ is defined as
\begin{align}\label{eqn:SD}
  \SD(X,Y) \eqdef \frac{1}{2} \sum_{\omega \in \Omega} \left| \Pr[X=\omega]
    - \Pr[Y=\omega] \right|.
\end{align}

\subsection*{The Model}
\paragraph{The communication model and distributed tasks.}
We consider the synchronous model where a set of $n$ parties
$\party{1},\dots,\party{n}$ run protocols. Each protocol consists of
\emph{rounds} in which parties send messages. We assume the existence of a
global counter which synchronizes parties in between rounds (but they are
asynchronous within a round). The parties communicate via a broadcast channel.

The focus of this work is on selection protocols where parties do not have any
private inputs and their goal is to agree on a sequence of random bits.
Examples of such tasks are coin-flipping protocols, leader election protocols,
etc.

Throughout this paper, we restrict ourselves to public-coin protocols.
\begin{definition}[Public-coin protocols]\label{def:public-coin}
  A protocol is \emph{public-coin} if all honest parties' messages consist only
  of uniform random bits.
\end{definition}
In \Cref{sec:public_coin} we argue that the restriction to public-coin protocols
is without loss of generality since in the full information model any selection
protocol can be converted into a public-coin one, without increasing the round
complexity and without degrading security (though this transformation may
significantly increase the communication complexity).

\paragraph{The adversarial model.}
We consider the {\em full information model} where it is assumed the adversary
is all powerful, and may see the entire transcript of the protocol.  The most
common adversarial model considered in the literature is the Byzantine model,
where a bound~$t=t(n)\leq n$ is specified, and the adversary is allowed to
corrupt up to~$t$ parties. The adversary can see the entire transcript, has full
control over all the corrupted parties, and can broadcast any messages on their
behalf. Moreover, the adversary has control over the order of the messages sent
within each round of the protocol.\footnote{Such an adversary is often referred
  to as ``rushing''.}  We focus on the Byzantine model throughout this work.

Within this model, two types of adversaries were considered in the literature:  {\em static} adversaries, who need to specify the parties they corrupt {\em before} the protocol begins, and {\em adaptive} adversaries, who can corrupt the parties {\em adaptively} based on the transcript so far.  Our results hold for both types of adversaries.  Throughout this work, we focus on the adaptive setting, since the proof is more complicated in this setting.  In Subsection~\ref{sec:static} we mention how to modify (and simplify) the proof for the static setting.

\paragraph{Correctness and security.}
For any protocol $\Pi$ and any adversary $\Adv$, we denote by
\begin{align*}
  \out{\Pi}{\Adv}{\sh_1,\dots,\sh_{n}}
\end{align*}
the output of the protocol $\Pi$ when executed with the adversary $\Adv$, and where each honest party~$\party{i}$ uses randomness~$\sh_i$.

Let $\Pi$ be a protocol whose output is a string in $\{0,1\}^m$ for some
$m\in\N$. Loosely speaking, we say that an adversary is ``successful'' if he manages to bias the output of the protocol to his advantage.
More specifically,
we say that an adversary is ``successful'' if he chooses a predetermined subset $M\subseteq\{0,1\}^m$ of some size~$s$, and succeeds in biasing the outcome towards the set~$M$.  To this end, for any set size~$s$, we define
\begin{align*}
  \valset{\Pi}{\Adv} &\eqdef \max_{M\subseteq\B^m\mbox{ s.t. }|M|=s}\valb{M}{\Pi}{\Adv}\\
  &\eqdef \max_{M\subseteq\B^m\mbox{ s.t. }|M|=s}
 \left( \Pr_{\sh_1,\dots,\sh_{n}}[\out{\Pi}{\Adv}{\sh_1,\dots,\sh_{n}} \in M]-\Pr_{\sh_1,\dots,\sh_{n}}[{\sf out}_\Pi({\sh_1,\dots,\sh_{n}}) \in M]\right),
\end{align*}
where ${\sf out}_\Pi({\sh_1,\dots,\sh_{n}})$ denotes the outcome of the protocol~$\Pi$ if all the parties are honest, and use randomness ${\sh_1,\dots,\sh_{n}}$.

Intuitively, the reason we parameterize over the set size~$s$ is that we may hope for different values of $\valb{M}{\Pi}{\Adv}$ for sets $M$ of different sizes, since for a large set $M$ it is often the case that
$\Pr_{\sh_1,\dots,\sh_{n}}[{\sf out}_\Pi({\sh_1,\dots,\sh_{n}}) \in M]$ is large, and hence $\valb{M}{\Pi}{\Adv}$ is inevitably small, whereas for small sets $M$ the value $\valb{M}{\Pi}{\Adv}$ may be large.

For example, for coin-flipping protocols (where $m=1$ and the outcome is a
uniformly random bit in the case that all parties are honest), often an
adversary is considered successful if it biases the outcome to his preferred bit
with probability close to~$1$, and hence an adversary is considered successful
if $\valb{M}{\Pi}{\Adv}\geq\frac{1}{2}-o(1)$ for either $M=\{0\}$ or
$M=\{1\}$, whereas for general selection protocols (where $m$ is a parameter)
one often considers subsets $M\subseteq\{0,1\}^m$ of size $\gamma\cdot2^m$ for
some constant $\gamma>0$, and an adversary is considered successful if there
exists a constant $\delta>0$ such that $\valb{M}{\Pi}{\Adv}\geq\delta$.

\begin{definition}[Security]\label{def:sec}
  Fix any constant $\delta>0$, any $t=t(n)\leq n$, and any $n$-party protocol
  $\Pi$ whose output is an element in $\B^m$.  Fix any $s=s(m)$.  We say that
  $\Pi$ is \emph{$(t,\delta,s)$-adaptively secure} if
  for any adversary $\Adv$ that adaptively
  corrupts up to $t=t(n)$ parties, it holds that
  \begin{align*}
    \valset{\Pi}{\Adv} \leq\delta.
  \end{align*}
\end{definition}
We note that this definition generalizes the standard security definition for coin-flipping protocols and selection protocols.
We emphasize that our results are quite robust to the specific security definition that we consider, and we could have used alternative definitions as well. Intuitively, the reason is that we  show how to transform any $d$-round protocol~$\Pi$ into another $d$-round protocol with short messages, that simulates~$\Pi$ (see Definition~\ref{def:sim} below), where this transformation is {\em independent} of the security definition.  Then, in order to prove that the resulting protocol is as secure as the original protocol~$\Pi$, we show that if there exists an adversary for the short protocol that manages to break security according to some definition, then there exists an adversary for~$\Pi$ that ``simulates'' the adversary of the short protocol and breaches security in the same way.
(See Section~\ref{sec:techniques} for more details, and Section~\ref{sec:proof} for the formal argument).

Finally, we mention that an analogous definition to Definition~\ref{def:sec} can be given for static adversaries.  Our results hold for the static definition as well.

\begin{definition}[Simulation]\label{def:sim}
Let $\Pi$ be an $n$-party protocol with outputs in $\{0,1\}^m$.  We say that an $n$-party protocol~$\Pi'$ simulates~$\Pi$ if
$$
\SD\left({\sf out}_\Pi,{\sf out}_{\Pi'}\right)=\negl(n),
$$
where ${\sf out}_\Pi$ is a random variable that corresponds to the output of
protocol $\Pi$ assuming all parties are honest, and ${\sf out}_{\Pi'}$ is a random variable that corresponds to the output of protocol $\Pi'$ assuming all parties are honest.
\end{definition}

\subsection*{Probabilistic Tools}
In the analysis we will use the following simple claims.

\begin{claim}\label{claim:prob}
  Let $k,M\in\mathbb{N}$ be two integers. Let $U\subseteq\{0,1\}^k$ and $f\colon
  U\rightarrow
  [M]$.  For every $i\in[M]$, denote
  by $$\alpha_i=\Pr_{u\leftarrow U}\left[f(u)=i\right].$$ Then,
  \begin{align*}
    \E_{u\leftarrow U}\left[\alpha_{f(u)}\right]\geq\frac{1}{M},
  \end{align*}
  and for any $\epsilon>0$,
  \begin{align*}
    \Pr_{u\leftarrow U}\left[\alpha_{f(u)}\geq \frac{\epsilon}{M}\right]\geq
    1-\epsilon.
  \end{align*}
\end{claim}
\begin{proof}
  We begin with the proof of the first part. By the definition of expectation
  \begin{align*}
    \E_{u\leftarrow U}\left[\alpha_{f(u)}\right] & =
    \sum_{u\in U}\Pr[U=u]\cdot \alpha_{f(u)}=
    \sum_{i=1}^M\alpha_i\cdot\Pr_{u\leftarrow
      U}\left[\alpha_{f(u)}=\alpha_i\right] \geq \sum_{i=1}^{M} \alpha_i^2.
  \end{align*}
 This, together with the the Cauchy-Schwarz inequality, implies that
  \begin{align*}
    \E_{u\leftarrow U}\left[\alpha_{f(u)}\right] \geq \sum_{i=1}^{M} \alpha_i^2 &=
    \sum_{i=1}^{M} \alpha_i^2 \cdot \sum_{i=1}^M \(\frac{1}{\sqrt{M}}\)^2 \\ &\geq
    \(\sum_{i=1}^M \alpha_i \cdot \frac{1}{\sqrt{M}}\)^2 = \frac{1}{M},
  \end{align*}
  where the last equality follows from the fact that $\sum_{i=1}^M \alpha_i =
  1$.

  For the second part, let
  \begin{align*}
    B=\left\{i\in[M]\mid \alpha_i < \frac{\epsilon}{M}\right\}.
  \end{align*}
  Then,
  \begin{align*}
    \Pr_{u\leftarrow U}\left[\alpha_{f(u)} < \frac{\epsilon}{M}\right] =
    \Pr_{u \in U}[f(u)\in B]\leq \sum_{i\in B}\alpha_i\leq
    |B|\cdot\frac{\epsilon}{M}\leq \epsilon,
  \end{align*}
  as desired, where the first inequality follows from the union bound and the
  definition of $\alpha_i$, the second
  inequality follows from the definition of $B$, and the third inequality follows
  from the fact that $|B|\leq M$.
\end{proof}

\begin{definition}[Entropy]\label{def:entropy}
  Let $X$ be a random variable with finite support. The (Shannon) entropy of $X$
  is defined as
  \begin{align*}
    \ent(X) = \sum_{x\in\supp(X)} \Pr[X=x] \cdot \log \frac{1}{\Pr[X=x]} =
    \E_{x\leftarrow X}\left[ \log \frac{1}{\Pr[X=x]}\right].
  \end{align*}
\end{definition}

\begin{claim}\label{claim:ent2sd}
  Let $X$ be a random variable with domain $\B^k$. If $\ent(X) \geq
  k-\eps$, then
  \begin{align*}
    \SD(X, \U_k) \leq \sqrt{\frac{\eps}{2}},
  \end{align*}
  where $\U_k$ is the uniform distribution over $k$ bits, and where $\SD(X, \U_k)$ denotes the statistical distance between $X$ and $\U_k$ (see Equation~\eqref{eqn:SD} for the definition of statistical distance).
\end{claim}
\newcommand{\RE}{\mathbf{D}_{\mathsf{KL}}}
\begin{proof}
  The \emph{relative entropy} (\aka the Kullback-Leibler divergence) between two
  distributions $\D_1,\D_2\subseteq \B^k$ is defined as
  \begin{align*}
    \RE(\D_1\|\D_2) = \sum_{x\in\B^k}\D_1(x)\cdot \log\( \frac{\D_1(x)}{\D_2(x)}
    \).
  \end{align*}
  A well known relation between relative entropy and the statistical distance is
  known as Pinsker's inequality which states that for any two distributions
  $\D_1,\D_2$ as above, it holds that
  \begin{align}\label{eq:pinsker}
    \SD(\D_1,\D_2) \leq \sqrt{\frac{\ln 2}{2}\cdot \RE(\D_1\|\D_2)}.
  \end{align}

  Thus, it remains to bound the relative entropy of $X$ and $\U_k$. Let $p_x =
  \Pr_{x\in\B^k}[X=x]$. We get that
  \begin{align*}
    \RE(X\|\U_k) & = \sum_{x\in\B^k}p_x\cdot \log\( p_x\cdot 2^{k} \)\\
    & = \sum_{x\in\B^k}p_x\cdot \(\log (p_x) + k \)\\
    & = -\ent(X) + k.
  \end{align*}
  Since $\ent(X) \geq k-\eps$, we get that
  \begin{align*}
    \RE(X\|\U_k) & \leq -k + \eps +k = \eps.
  \end{align*}
  Plugging this into Pinsker's inequality (see \Cref{eq:pinsker}), we get that
  \begin{align*}
    \SD(X, \U_k) \leq \sqrt{\frac{\ln 2}{2}\cdot \eps}\leq \sqrt{\frac{\eps}{2}}.
  \end{align*}
\end{proof}

\section{Compressing Communication in Distributed Protocols}\label{sec:proof}
In this section we show how to transform any $n$-party $d$-round $t$-adaptively secure
public-coin protocol, that outputs messages of length $m$ and sends messages of
length $L$, into an $n$-party $d$-round $t$-adaptively secure
public-coin protocol in which every party sends messages of length $\ell =
m\cdot\polylog(n,d)$.

Throughout this section, we fix $\mu^*$ to be the negligible function defined by
\begin{equation}\label{eqn:mu*}
  \mu^*=\mu^*(n,d)=\left(\sqrt{\epsilon}+1-(1-\epsilon)^{dn}\right)\cdot 2dn,
\end{equation}
  and where $\epsilon=2^{-\log^2(dn)}$.

\begin{theorem}\label{thm:main}
Fix any $m=m(n)$, $d=d(n)$, $L=L(n)$, and any $n$-party $d$-round public-coin
  selection protocol $\Pi$ that outputs messages in $\{0,1\}^m$ and in which all parties send
  messages of length $L=L(n)$.  Then, for any constant $\delta>0$, any $t=t(n)<n$, and any $s=s(m)$, if $\Pi$ is $(t,\delta,s)$-adaptively secure then there exists an $n$-party $d$-round
  $(t,\delta',s)$-adaptively secure public-coin selection protocol, that simulates~$\Pi$, where all parties send messages of length $\ell =
  m\cdot\log^4(n\cdot d)$, and where $\delta'\leq\delta+\mu^*$ (and $\mu^*=
\mu^*(n,d)$ is the negligible function defined in Equation~\eqref{eqn:mu*}).
\end{theorem}

 \begin{proof}

Fix any  $m=m(n)$, $d=d(n)$, $L=L(n)$,  and any $n$-party $d$-round public-coin
  protocol $\Pi$ that outputs messages in $\{0,1\}^m$ and in which all parties send
  messages of length $L=L(n)$. Fix any constant $\delta>0$, any $t=t(n)<n$, and any $s=s(m)$ such that $\Pi$ is $(t,\delta,s)$-adaptively secure. We start by describing the construction of the (short message) protocol. Let
\begin{equation}\label{eqn:N}
N=2^\ell=2^{m\cdot\log^4(n\cdot d)}.
\end{equation}
Let $$\Hc = \{H:[d\cdot n]\times\B^\ell\to\B^L\}$$ be the set all
possible $[d\cdot n]\times \B^\ell\equiv [d\cdot n]\times [N]$ matrices, whose elements are from $\B^L$.
Note that $|\Hc|= 2^{d\cdot n\cdot N\cdot L}$.   We often
interpret $H:[d\cdot n]\times\B^\ell\to\B^L$ as a function
$$H:[d]\times
[n]\times\B^\ell\to\B^L,
$$
or as a matrix where each row is described by a pair from $[d]\times[n]$.  We abuse notation and denote by
$$H(i,j,\sh)\eqdef H((i-1)n+j,\sh).$$
As a convention, we denote by $\lo$ a
message from $\B^L$ and by $\sh$ and a message from $\B^\ell$. \\

From now on, we assume for the sake of simplicity of notation, that in protocol~$\Pi$, in each round, all the parties send a message.
Recall that we also assume for the sake of simplicity (and without loss of
generality) that~$\Pi$ is a public-coin protocol (see
Definition~\ref{def:public-coin}). For any $H\in\Hc$ we define a protocol $\Pi_H$ that simulates the
execution of the protocol $\Pi$, as follows.

\paragraph{The Protocol~$\Pi_H$.}
In the protocol $\Pi_H$, for every $i\in[d]$ and $j\in[n]$,
in the $\ith{i}$ round, party $\party{j}$ sends a random string $\sh_{i,j}\leftarrow\{0,1\}^\ell$.
We denote the resulting transcript in round $i$ by
$$
\trans_{H,i}=(\sh_{i,1},\ldots,\sh_{i,n})\in\left(\{0,1\}^{\ell}\right)^n,
$$
and denote the entire transcript  by
$$\trans_H=(\trans_{H,1}\ldots,\trans_{H,d}).
$$
We abuse notation, and define for every round $i\in[d]$,
$$
H(\trans_{H,i})=(H(i,1,\sh_{i,1}),\ldots,H(i,n,\sh_{i,n})).
$$
Similarly, we define
$$
H(\trans_H)=(H(\trans_{H,1})\ldots,H(\trans_{H,d})).
$$
The outcome of protocol~$\Pi_H$ with transcript $\trans_H$ is defined to be the outcome of protocol~$\Pi$ with transcript $H(\trans_H)$.\\

It is easy to see that the round
complexity of $\Pi_H$ (for every $H\in\Hc$) is the same as that of~$\Pi$.  Moreover, we note that with some complication in notation we could have also preserved the exact communication pattern (instead of assuming that in each round all parties send a message).

In order to prove \Cref{lemma:many2few} it suffices to prove the following two lemmas.
\begin{lemma}\label{lemma:H_secure}
  There exists a subset $\Hc_0\subseteq\Hc$ of size $\frac{\abs{\Hc}}{2}$, such that for every matrix $H\in\Hc_0$ it holds that $\Pi_H$ is $(t,\delta',s)$-adaptively secure for $\delta'= \delta+\mu^*$, where $\mu^*$ is the negligible function defined in Equation~\eqref{eqn:mu*}.

\end{lemma}

\begin{lemma}\label{claim:simulation}
There exists a negligible function $\mu=\mu(n,d)$ such that,
$$
\Pr_{H\leftarrow \Hc}[\SD({\sf out}_{\Pi_H},{\sf out}_{\Pi})\leq\mu]\geq\frac{2}{3}.
$$
\end{lemma}

Indeed, given \Cref{lemma:H_secure,claim:simulation}, we obtain that there
exists an $H\in\Hc$ such that $\Pi_H$ is $(t,\delta', s)$-adaptively secure and
it simulates $\Pi$.
\end{proof}

In \Cref{sec:claim_simulation} we give the proof of \Cref{claim:simulation} and
in \Cref{sec:proof_of_lemma} we give the proof of \Cref{lemma:H_secure}.

\subsection{Proof of Lemma~\ref{claim:simulation}}\label{sec:claim_simulation}
  By the definition of statistical distance, in order to prove
  Lemma~\ref{claim:simulation} it suffices to prove that there exists a
  negligible function $\mu=\mu(n,d)$ such that,
$$
\Pr_{H\leftarrow \Hc}\left[\forall z\in\{0,1\}^m, \abs{\Pr[{\sf
      out}_{\Pi_H}=z]-\Pr[{\sf
      out}_{\Pi}=z]}\leq\frac{\mu}{2^m}\right]\geq\frac{2}{3}.
$$

Note that
\begin{align*}
  &\Pr_{H\leftarrow \Hc}\left[\forall z\in\{0,1\}^m, \abs{\Pr[{\sf out}_{\Pi_H}=z]-\Pr[{\sf out}_{\Pi}=z]}\leq\frac{\mu}{2^m}\right]=\\
  &1-\Pr_{H\leftarrow \Hc}\left[\exists z\in\{0,1\}^m, \abs{\Pr[{\sf out}_{\Pi_H}=z]-\Pr[{\sf out}_{\Pi}=z]}>\frac{\mu}{2^m}\right]\geq\\
  &1-\sum_{z\in\{0,1\}^m}\Pr_{H\leftarrow \Hc}\left[\abs{\Pr[{\sf
        out}_{\Pi_H}=z]-\Pr[{\sf out}_{\Pi}=z]}>\frac{\mu}{2^m}\right].
\end{align*}
Therefore, it suffices to prove that there exists a negligible function $\mu$
such that for every $z\in\{0,1\}^m$,
$$
\Pr_{H\leftarrow \Hc}\left[\abs{\Pr[{\sf out}_{\Pi_H}=z]-\Pr[{\sf
      out}_{\Pi}=z]}>\frac{\mu}{2^m}\right]\leq\frac{1}{3\cdot 2^m}.
$$
To this end, for any $z\in\{0,1\}^m$, we denote by $p_z=\Pr[{\sf out}_{\Pi}=z]$
and $p_{z,H}=\Pr[{\sf out}_{\Pi_H}=z]$. Using this notation, it suffices to
prove that there exists a negligible function $\mu$ such that for every
$z\in\{0,1\}^m$,
$$
\Pr_{H\leftarrow
  \Hc}\left[\abs{p_{z,H}-p_z}>\frac{\mu}{2^m}\right]\leq\frac{1}{3\cdot 2^m}.
$$

For any $H\in\Hc$, consider the experiment, where we run the protocol $\Pi_H$
independently $B=2^{m\cdot\log^3(nd)}$ times, and check how many times the
output is~$z$. Denote by $X_1,\ldots,X_B$ the identically distributed random
variables, where $X_i=1$ if in the $\ith{i}$ run of the protocol the outcome
is~$z$, and $X_i=0$ otherwise.  The Chernoff bound\footnote{The Chernoff bound
  states that for any identical and independent random variables
  $X_1,\ldots,X_B$, such that $X_i\in\{0,1\}$ for each~$i$, if we denote by
  $p=\E[X_i]$ then $\Pr[\left|\frac{1}{B}\sum_{i=1}^B
    X_i-p\right|\geq\delta]\leq e^{-\frac{\delta^2}{3} B}$.} implies that for
every $H\in\Hc$ and for every $\gamma>0$,
$$
\Pr\left[\left|\frac{1}{B}\sum_{i=1}^B X_i- p_{z,H}\right|\geq \gamma\right]\leq
e^{-\frac{\gamma^2\cdot B}{3}}.
$$
In particular, setting $\gamma=2^{-m\cdot \log^2(nd)}$ we deduce that
\begin{equation}\label{eqn:X}
  \Pr\left[\left|\frac{1}{B}\sum_{i=1}^B X_i-p_{z,H}\right|\geq \gamma\right]\leq e^{-2^{m\cdot \log^2(nd)}}.
\end{equation}

We next define random variables $Y_1,\ldots,Y_B$ as follows: We run the protocol
$\Pi$ independently $B$ times, and we set $Y_i=1$ if in the $\ith{i}$ run the
outcome is~$z$, and otherwise we set $Y_i=0$.  We note that the same argument
used to deduce Equation~\eqref{eqn:X} can be used to deduce that
\begin{equation}\label{eqn:Y}
  \Pr\left[\left|\frac{1}{B}\sum_{i=1}^B Y_i- p_{z}\right|\geq\gamma\right]\leq e^{-2^{m\cdot \log^2(nd)}}.
\end{equation}
Note that,
\begin{align*}
  &\Pr\left[\left|p_{z,H}-p_z\right|>4\gamma\right]\leq\\
  &\Pr\left[\left|p_{z,H}-\frac{1}{B}\sum_{i=1}^B X_i\right|+ \left|\frac{1}{B}\sum_{i=1}^B X_i-\frac{1}{B}\sum_{i=1}^B Y_i\right|+ \left|\frac{1}{B}\sum_{i=1}^B Y_i-p_z\right|>4\gamma\right]\leq\\
  &\Pr\left[\left|p_{z,H}-\frac{1}{B}\sum_{i=1}^B X_i\right|>\gamma\right]+\Pr\left[\left|\frac{1}{B}\sum_{i=1}^B X_i-\frac{1}{B}\sum_{i=1}^B Y_i\right|>2\gamma\right]+\Pr\left[\left|\frac{1}{B}\sum_{i=1}^B Y_i-p_z\right|>\gamma\right]\leq\\
  &2\cdot e^{-2^{m\cdot \log^2(nd)}}+\Pr\left[\left|\frac{1}{B}\sum_{i=1}^B
      X_i-\frac{1}{B}\sum_{i=1}^B Y_i\right|>2\gamma\right],
\end{align*}
where the first inequality follows from the triangle inequality, the second
inequality follows from the union bound, and the third inequality
follows from Equations~\eqref{eqn:X} and~\eqref{eqn:Y}.  Thus, it suffices to
prove that there exists a negligible function~$\mu=\mu(n,d)$ such that
$$
\Pr\left[\left|\frac{1}{B}\sum_{i=1}^B X_i-\frac{1}{B}\sum_{i=1}^B
    Y_i\right|>2\gamma\right]\leq\frac{\mu}{2^m}.
$$
To this end, notice that for a random $H\leftarrow\Hc$,
\begin{align*}
  &\SD\left(\left(X_1,\ldots,X_B\right),\left(Y_1,\ldots, Y_B\right)\right)\leq \\
  &\sum_{i=1}^B \SD\left(\left(X_1,\ldots,X_{i-1},X_i,Y_{i+1},\ldots,Y_B\right),\left(X_1,\ldots,X_{i-1},Y_i,Y_{i+1},\ldots,Y_B\right)\right)=\\
  &\sum_{i=1}^B \SD\left(\left(X_1,\ldots,X_{i-1},X_i\right),\left(X_1,\ldots,X_{i-1},Y_i\right)\right)\leq\\
  &B\cdot \SD\left(\left(X_1,\ldots,X_{B-1},X_B\right),\left(X_1,\ldots,X_{B-1},Y_B\right)\right)\leq\\
  &B\cdot nd\cdot\frac{(B-1)nd}{Nnd}\leq \\
  &\frac{B^2\cdot nd}{N}\leq\\
  &\frac{2^{2m\log^3(nd)}\cdot nd}{2^{m\log^4(nd)}}\leq\\
  &2^{-m\log^3(nd)},
\end{align*}
where the first equation follows from a standard hybrid argument. The second
equation follows from the fact that $Y_{i+1},\ldots,Y_B$ are independent of
$X_1,\ldots,X_i,Y_i$.  The third equation follows from the fact that the
statistical distance between $(X_1,\ldots,X_{i-1},X_i)$ and
$(X_1,\ldots,X_{i-1},Y_i)$ is maximal for $i=B$.  The forth equation follows
from the fact that $(X_1,\ldots,X_{B-1},X_B)$ and $(X_1,\ldots,X_{B-1},Y_B)$ are
identically distributed if the following event, which we denote by $\GOOD$,
occurs: Recall that each $X_i$ depends only on $nd$ {\em random} coordinates of
$H\leftarrow\Hc$. We say that $\GOOD$ occurs if the $nd$ coordinates that $X_B$
depends on are disjoint from all the $nd(B-1)$ coordinates that
$X_1,\ldots,X_{B-1}$ depend on. The forth equation follows from the fact that
$\Pr[\neg{\GOOD}]\leq nd\cdot\frac{(B-1)nd}{Nnd}$. The rest of the equations
follow from basic arithmetics and from the definition of~$B$ and~$N$.

In particular, this implies that
\begin{equation}\label{eqn:SDXY}
  \SD\left(\left(\frac{1}{B}\sum_{i=1}^B X_i\right), \left(\frac{1}{B}\sum_{i=1}^B Y_i\right)\right)\leq 2^{-m\log^3(nd)}.
\end{equation}

Consider the algorithm ${\cal D}$ that given $p'_z$, supposedly
distributed according to $\frac{1}{B}\sum_{i=1}^B X_i$ or distributed according
to $\frac{1}{B}\sum_{i=1}^B Y_i$, outputs~$1$ if $|p'_z-p_z|\leq \gamma$, and
otherwise outputs~$0$.  Equation~\eqref{eqn:Y} implies that
$$
\Pr\left[{\cal D}\left(\frac{1}{B}\sum_{i=1}^B Y_i\right)=1\right]\geq
1-e^{-2^{m\cdot \log^2(nd)}}.
$$
This together with Equation~\eqref{eqn:SDXY}, implies that
$$
\Pr\left[{\cal D}\left(\frac{1}{B}\sum_{i=1}^B X_i\right)=1\right]\geq
1-e^{-2^{m\cdot \log^2(nd)}}-2^{-m\log^3(nd)}\geq 1-2^{-m\log^2(nd)},
$$
which by the definition of ${\cal D}$, implies that
$$
\Pr\left[\left|\frac{1}{B}\sum_{i=1}^B X_i-p_z\right|\leq \gamma\right]\geq 1-
2^{-m\log^2(nd)}.
$$
This, in particular, implies that
$$
\Pr\left[\left|\frac{1}{B}\sum_{i=1}^B X_i-\frac{1}{B}\sum_{i=1}^B
    Y_i\right|\leq 2\gamma\right]\geq 1-2\cdot 2^{-m\log^2(nd)},
$$
as desired.

\subsection{Proof of Lemma \ref{lemma:H_secure}}\label{sec:proof_of_lemma}
Assume towards contradiction that for every set $\Hc_0\subseteq\Hc$ of size $\frac{\abs{\Hc}}{2}$ there exists $H\in\Hc_0$ such that $\Pi_H$ is
\emph{not} $(t,\delta',s)$-adaptively secure, for $\delta'= \delta+\mu^*$. This implies that there exists a set $\Hc_0\subseteq\Hc$ of size $\frac{\abs{\Hc}}{2}$ such that for every $H\in\Hc_0$ there exists an
adversary $\Adv^H$ that adaptively corrupts at most~$t$ parties and satisfies $$\valset{\Pi_H}{\left(\Adv^H\right)} \geq \delta'.$$

This, in turn, implies that there exists a set $M\subseteq\{0,1\}^m$ of size~$s>0$
such that for at least $1/{{2^m}\choose{s}}$-fraction of the
$H$'s in $\Hc_0$ the adversary $\Adv^H$ satisfies that $\valb{M}{\Pi_H}{\left(\Adv^H\right)}\geq\delta'$.  We denote this
set of $H$'s by $\Hc_1$. Notice that
\begin{align}\label{eqn:H1}
  \abs{\Hc_1} \geq \frac{\abs{\Hc_0}}{{2^m \choose s}} = \frac{\abs{\Hc}}{2\cdot{2^m
      \choose s}} \geq \frac{\abs{\Hc}}{2^{2^m}} = 2^{dnNL-2^m}.
\end{align}

The proof proceeds as follows: we show how to use these adversaries
$\{\Adv^H\}_{H\in \Hc_1}$ to construct an adversary $\Adv$ such that
$$\valb{M}{\Pi}{\Adv}\geq\delta'-\mu^*/2=\delta+\mu^*-\mu^*/2>\delta,$$
contradicting the $(t,\delta,s)$-adaptive security of~$\Pi$.

The idea is for the adversary $\Adv$ to simulate the execution of one of the
$\Adv^H$'s. The problem is that we do not know ahead of time which $H$ will be
consistent with the transcript of the protocol, since we have no control over the (long) random messages of the honest parties.
We overcome this problem by choosing $H$ {\em adaptively}.  Namely, at any point in the protocol, $\Adv$
simulates a random adversary $\Adv^H$, where $H$ is a random matrix that
is consistent (in some sense that we explain later) with the transcript up to that point.

More specifically, for every $i\in[d]$ and every $j\in[n]$, we denote by
$\Hc_{i,j-1}$ the set of matrices that are consistent with the transcript up
until the point where the $\ith{j}$ message of the $\ith{i}$ round is about to
be sent.  Fix any round $i\in[d]$ and any $j\in[n]$. Roughly speaking, in the $\ith{i}$ round before the $\ith{j}$ message is to be sent, the adversary $\Adv$
simulates $\Adv^{H^*}$ where $H^*\leftarrow\Hc_{i,j-1}$ is chosen
uniformly at random. If $\Adv^{H^*}$ corrupts a party $\party{u}$ then $\Adv$
also corrupts $\party{u}$.  If $\Adv^{H^*}$ sends a message $\sh_i^*$ on behalf
of a corrupted party $\party{u}$, then $\Adv$ will send the message
$\lo^*_i=H^*(i,u, \sh_i^*)$ on behalf of party $\party{u}$.  In this case, we
define $\Hc_{i,j}$ to be all the matrices in $\Hc_{i,j-1}$ which are consistent
with the transcript so far and agree with $H^*$ on row
$(i,u)$.  If $\Adv^{H^*}$ asks an honest party $\party{u}$ to send its message,
the adversary~$\Adv$ will also ask honest party $\party{u}$ to send a message. Upon
receiving a message $\lo^*$ from $\party{u}$, we choose a random matrix $H\leftarrow\Hc_{i,j-1}$ that is consistent with the transcript so far, and set $\Hc_{i,j}$ to be all the
matrices in $\Hc_{i,j-1}$ that are consistent with the transcript so far, and
where we fix the $(i,u)$ row to be the $(i,u)$ row of~$H$.

Before giving the precise description of the adversary $\Adv$, we provide some
useful notation.  We denote the transcript generated in an execution of the
protocol~$\Pi$ with an adversary $\Adv$ by $\trans_{\real}$. Note that
$\trans_{\real}$ consists of $d$ vectors (one per each round), where each vector
consists of~$n$ pairs of the form
\begin{align*}
  ((\party{j_1},\lo_1),\ldots,(\party{j_{n}},\lo_n)),
\end{align*}
where $\lo_1,\ldots\lo_n\in\B^L$ and $j_1,\dots,j_n\in[n]$, where the order
means that in this round party $\party{j_1}$ sent his message first, then party
$\party{j_2}$ sent his message, and so on (recall that in our model, the adversary has control over the scheduling of the messages within each round). We sometimes consider a partial
transcript $\trans_{i,j}$ (\ie a prefix of a transcript) which corresponds to a
partial execution of the protocol~$\Pi$ with the adversary~$\Adv$ until after the $\ith{j}$ message in the
$\ith{i}$ round was sent. For $H\in\Hc$, we denote by
\begin{align*}
  \MAP_H\colon [d]\times [n]\times\B^L \to \B^\ell \cup \left\{\bot\right\}
\end{align*}
the mapping that takes as input a row number $(i,j)\in[d]\times[n]$ and a (long) message
in $\lo\in \B^L$, and converts it into a (short) message $\sh\in\B^\ell$ such
that $H(i, j, \sh) = \lo$.  If no such message exists, $\MAP_H$ outputs~$\bot$.

Let $\trans_{i,j}$ be a (long) partial transcript of $\Pi$. The corresponding
(short) transcript of $\Pi_H$, denoted by $\MAP_H(\trans_{i,j})$, is defined
recursively, as follows.  Let $\trans_{i,j}=(\trans_{i,j-1}, (\party{u},
\lo))$.  Then,
\begin{align*}
  \MAP_H(\trans_{i,j}) = \big(\MAP_H(\trans_{i,j-1}), (\party{u}, \MAP_H(i, u,
  \lo))\big).
\end{align*}
We initialize $\trans_{1,0}=\emptyset$ and $\Hc_{1,0}=\Hc_1$.  Using this
notation, a formal description of the adversary $\Adv$ is given in
Figure~\ref{fig:adv_long2short}. \\

\begin{algorithm}
  \caption{The adversary $\Adv$ before the $\ith{j}$ message of round
    $i$.} \label{fig:adv_long2short}%

  \center{\bf The adversary $\boldsymbol{\Adv(\trans_{i,j-1})}$ before the
    $\ith{j}$ message of round $i$ }
  \begin{MyEnumerate}
  \item If $\Hc_{i,j-1} = \emptyset$, output $\bot$ and HALT.
  \item\label{it:2} Choose $H^*\leftarrow \Hc_{i,j-1}$ uniformly at random.  Let
    $\trans_{H^*}=\MAP_{H^*}(\trans_{i,j-1})$ denote the (short) transcript in
    the protocol $\Pi_{H^*}$ that corresponds to the (long) transcript
    $\trans_{i,j-1}$.
  \item\label{it:3} If $\Adv^{H^*}(\trans_{H^*})$ corrupts a party $\party{u}$
    then corrupt $\party{u}$.
  \item\label{it:4} If $\Adv^{H^*}(\trans_{H^*})$ sends a message on behalf of a
    corrupt party $\party{u}$, then do the following:
    \begin{MyEnumerate}{}
    \item Denote by $\sh^*\in\{0,1\}^\ell$ the message that
      $\Adv^{H^*}(\trans_{H^*})$ sends on behalf of $\party{u}$. Let $\lo^* =
      H^*(i,u,\sh^*)$.
    \item Send the message $\lo^*$ on behalf of party $\party{u}$.
    \item Add $(\party{u},\lo^*)$ to the partial transcript.  Namely, set
      \begin{align*}
        \trans_{i,j}=\(\trans_{i,j-1},(\party{u},\lo^*)\).
      \end{align*}

    \item\label{it:4d} Define $\Hc_{i,j}$ to be the set of all $H\in\Hc_{i,j-1}$
      that are consistent with the transcript so far, and for which
      $H(i,u,\cdot) = H^*(i,u,\cdot)$. Namely, set
        \begin{align*}
          \Hc_{i,j} = \big\{ H \in \Hc_{i,j-1}\mid\;& \forall \sh\colon\;
          H(i,u,\sh) = H^*(i,u,\sh), \text{and} \\
          & \text{$\Adv^H(\trans_H)$ sends $\sh^*$ on behalf of $\party{u}$},
            \\
            & \quad \text{where $\trans_H=\MAP_H(\trans_{i,j-1})$} \big\}.
        \end{align*}
      \end{MyEnumerate}
    \item\label{it:5} If $\Adv^{H^*}(\trans_{H^*})$ does not corrupt, and orders
      an honest party $\party{u}$ to send a message, then do the following:
    \begin{enumerate}
    \item Do not corrupt, and order honest party $\party{u}$ to send a
      message. Denote the message it sends by $\lo^*$.
     \item Add $(\party{u},\lo^*)$ to the partial transcript.  Namely, set
      \begin{align*}
        \trans_{i,j}=(\trans_{i,j-1},(\party{u},\lo^*)).
      \end{align*}
    \item Choose a random matrix
      \begin{align*}
        H'\leftarrow \{H\in \Hc_{i,j-1} \mid\; & \Adv^H(\trans_H) \text{ orders
          honest
          $\party{u}$ to send a message, and } \\
        & \exists \sh \text{ s.t. }H(i,u,\sh)=\lo^* \}.
      \end{align*}

    \item\label{it:5c} Define $\Hc_{i,j}$ to be the set of all $H\in
      \Hc_{i,j-1}$ that are consistent with the transcript so far, and agree
      with $H'$ on row $(i,u)$. That is,
      \begin{align*}
        \Hc_{i,j} = \{ H \in \Hc_{i,j-1} \mid\;& \forall \sh\colon\;
        H(i,u,\sh) = H'(i,u,\sh), \text{ and }\\
        & \Adv^H(\trans_H) \text{ orders honest
          $\party{u}$ to send a message}\}.
      \end{align*}
       \end{enumerate}
  \item\label{it6} If $j=n$, set $\Hc_{i+1,0} = \Hc_{i,j}$ and $\trans_{i+1,0} =
    \trans_{i,j}$.
  \end{MyEnumerate}
\end{algorithm}

In order to prove \Cref{lemma:H_secure} (and thus to complete the proof of
\Cref{lemma:many2few}), it suffices to prove the following lemma.
\begin{lemma}\label{lemma:reformulated}
  The adversary $\Adv$ makes at most $t$ adaptively-chosen corruptions, and $\valb{M}{\Pi}{\Adv}\geq\delta'-\mu^*/2$.
\end{lemma}
\begin{proof}
We first note that $\Adv$ always makes at most~$t$ corruptions. This follows from the fact that~$\Adv$ is always consistent with some adversary $\Adv^H$, for some $H\in\Hc_1$ (or else $\Adv$ aborts), and by our assumption, every $\Adv^H$ makes at most~$t$ corruptions.

We next prove that
$\valb{M}{\Pi}{\Adv}\geq\delta'-\mu^*/2$.
Recall that we denote by $\trans_\Adv$ the random variable that corresponds to the
  transcript generated by running the protocol $\Pi$ with the adversary~$\Adv$
  (described in Figure~\ref{fig:adv_long2short}).

  Let $\trans_{\ideal}$ be an ``ideal'' transcript, generated as follows: Choose
  a random $H\leftarrow\Hc_1$, run the protocol $\Pi_H$ with the adversary
  $\Adv^H$.  Denote the resulting transcript by $\trans_H$.  As above,
  $\trans_H$ consists of $d$ vectors (one per each round), where each vector
consists of $n$ pairs of the form
\begin{align*}
  ((\party{j_1},\sh_1),\ldots,(\party{j_{n}},\sh_n)),
\end{align*}
where $\sh_1,\ldots\sh_n\in\B^\ell$ and $j_1,\dots,j_n\in[n]$.  We define
  $$
  \trans_\ideal=H(\trans_H)
  $$
 where $H(\trans_H)$ is the transcript obtained by applying $H(i,u,\cdot)$ to each element in the $\ith{(i,u)}$ row of $\trans_H$. Formally, $H(\trans_H)$  is defined recursively, as follows: For every $i\in[d]$ and every $j\in[n]$, we let $\trans_{H,i,j}$ denote the transcript $\trans_H$ up until after the $\ith{j}$ message in the $\ith{i}$ round is sent.
 We define $H(\trans_{H,i,j})$ recursively, as follows:  For $\trans_{H,i,j}=(\trans_{H,i,j-1},(\party{u},\sh))$, we define
 $$
 H(\trans_{H,i,j})=(H(\trans_{H,i,j-1}),(\party{u},H(i,u,\sh))).
 $$

 In order to prove Lemma~\ref{lemma:reformulated} it suffices to prove the
 following claim.

 \begin{claim}\label{claim:SD}
   \begin{align*}
     \SD(\trans_{\real},\trans_{\ideal}) =\mu^*/2,
   \end{align*}
 \end{claim}

\begin{proof}
  We prove Claim~\ref{claim:SD} using a hybrid argument. Specifically, we define
  a sequence of $d\cdot (n+1)$ experiments.  For every $i\in[d]$ and every
  $j\in\{0,1,\ldots,n\}$, we define the experiment $\Exp^{(i,j)}$ as follows:
  \begin{enumerate}
  \item Generate $\trans_{i,j}$ and $\Hc_{i,j}$, as defined in
    Figure~\ref{fig:adv_long2short}.
  \item\label{it:12} Choose a random $H\leftarrow \Hc_{i,j}$, and let
    $\trans_{H,i,j}=\MAP_H(\trans_{i,j})$.
  \item\label{it:13} Run the protocol $\Pi_H$ with the adversary $\Adv^H$, given
    the partial transcript $\trans_{H,i,j}$.  Namely, run $\Pi_H$ with $\Adv^H$ from after the $\ith{j}$ message in the $\ith{i}$ round was sent, and assume the transcript up until that point is $\trans_{H,i,j}$.  Denote the entire transcript
    (including $\trans_{H,i,j}$) by $\trans_H$.
  \item\label{it:14} Output $H(\trans_H)$.
  \end{enumerate}
  Notice that
  \begin{align*}
  \Exp^{(d,n)}\equiv \trans_\real,
  \end{align*}
  and
  \begin{align*}
    \Exp^{(1,0)}\equiv \trans_\ideal.
  \end{align*}

  It remains to argue that for every $i\in[d]$ and every $j\in[n]$ the
  statistical distance between any two consecutive experiments $\Exp^{(i,j-1)}$
  and $\Exp^{(i,j)}$ is small. In particular, it suffices to prove that
  \begin{align}\label{eq:hybrid}
    \SD\left(\Exp^{(i,j-1)}, \Exp^{(i,j)}\right) =\frac{\mu^*}{2dn}.
  \end{align}
The reason is that
given this inequality, we obtain that
  \begin{align*}
    \SD(\trans_{\real},\trans_{\ideal}) \leq \sum_{i\in[d],j\in[n]} \SD(\Exp^{(i,j-1)},
    \Exp^{(i,j)}) \leq d\cdot n \cdot\frac{\mu^*}{2dn}=\frac{\mu^*}{2},
  \end{align*}
   which completes the claim.  We note that the first inequality follows from the union bound together with the fact that $\Exp^{(i,n)}=\Exp^{(i+1,0)}$ for every $i\in[d-1]$ (see Figure~\ref{fig:adv_long2short} Item~\ref{it6}).\\

  We proceed with the proof of \Cref{eq:hybrid}. To this end, fix any $i\in[d]$
  and $j\in[n]$.  Let $k\eqdef (i-1)\cdot d+j$.  Note that in both
  $\Exp^{i,j-1}$ and $\Exp^{i,j}$ the first $k-1$ messages are generated according
  to $\trans_{\Adv}$.

  Denote by $\corr$ the event that the $\ith{k}$ message is sent by a
  corrupted party.  We first argue that
  $$
  \Pr\left[\corr\mid\Exp^{(i,j-1)}\right]=  \Pr\left[\corr\mid\Exp^{(i,j)}\right].
  $$
  This follows immediately from the definition of the two experiments.  In $\Exp^{(i,j)}$ (according to Figure~\ref{fig:adv_long2short}, Items~\ref{it:2}-\ref{it:4}), before sending the $\ith{k}$ message, a random function is chosen $H^*\leftarrow\Hc_{i,j-1}$ and the $\ith{k}$ message is sent by a corrupted party if and only if $\Adv^{H^*}$ chooses the $\ith{k}$ message to be sent by a corrupted party (given the transcript so far).
  Note that in $\Exp^{(i,j-1)}$, the same exact process occurs (see
  \Cref{it:12,it:13,it:14} at the beginning of the proof of
  Claim~\ref{claim:SD}).

  We next argue
  \begin{equation}\label{eqn:Exp-corr}
    \SD\(\(\Exp^{(i,j-1)}\mid \corr\), \(\Exp^{(i,j)} \mid \corr\)\) = 0.
  \end{equation}
  To see why \Cref{eqn:Exp-corr} holds, note that according to
  Figure~\ref{fig:adv_long2short} (see \Cref{it:2,it:3,it:4}), the
  $\ith{k}$ message in $\left(\Exp^{(i,j)}\mid \corr\right)$ is chosen by
  sampling a random matrix $H^*\leftarrow\Hc_{i,j-1}$ conditioned on the fact
  that the $\ith{k}$ message sent in $\Pi_{H^*}$ with $\Adv^{H^*}$ is sent
  by a corrupted party. Denote this corrupted party by $\party{u}$ and denote by
  $\sh^*$ the message that $\Adv^{H^*}$ sends on behalf of $\party{u}$.  Then
  the $\ith{k}$ message in $\Exp^{(i,j)}$ is set to be $H^*(i,u,\sh^*)$.
  Note that the $\ith{k}$ message in $\Exp^{(i,j-1)}$ is chosen in exactly the same way (see
  \Cref{it:12,it:13,it:14} at the beginning of the proof of
  Claim~\ref{claim:SD}).  Moreover, the distribution of the set $\Hc_{i,j}$ in both cases is identical, which implies that the distributions of the rest of the messages in $\(\Exp^{(i,j-1)}\mid
  \corr\)$ and in $\(\Exp^{(i,j)}\mid \corr\)$ are identical as well.

  It remains to prove that
  \begin{align}\label{eqn:hybrid-corr}
    \SD\left(\left(\Exp^{(i,j-1)}\mid \neg{\corr}\right),
      \left(\Exp^{(i,j)}\mid\neg{\corr}\right)\right) = \frac{\mu^*}{2dn}.
  \end{align}

Recall that in $\left(\Exp^{(i,j)}\mid \neg{\corr}\right)$ the $\ith{k}$ message is uniformly distributed in $\B^L$. Denote by $R'$ the $\ith{k}$ message in $\left(\Exp^{(i,j-1)}\mid \neg{\corr}\right)$.  Recall that $R'$ is distributed as follows:  Choose a random $H\leftarrow\Hc_{i,j-1}$ such that the adversary $\Adv$ (given the partial transcript $\MAP_H(\trans_{i,j-1})$) orders an honest party $\party{u}$ to send the $\ith{j}$ message in the $\ith{i}$ round.  Choose a random $r'\leftarrow\{0,1\}^\ell$, and and set
$R'=H(i,u,r')$.

Notice that in order to prove Equation~\eqref{eqn:hybrid-corr}, it suffices to prove that
\begin{equation}\label{SD:R'}
\SD(R',\U_L)=\frac{\mu^*}{2dn}.
\end{equation}

Recall that we fixed $\epsilon=2^{-\log^2 (dn)}$.
We argue that in order to prove Equation~\eqref{SD:R'} it suffices to prove
that,
  \begin{equation}\label{eqn:sizeH}
    \Pr\left[\abs{\Hc_{i,j-1}}\geq\frac{2^{dnN L}}{2^{(k-1)NL}\cdot \left(\frac{4nN}{\epsilon}\right)^{k-1}\cdot 2^{2^m}}\right]\geq (1-\epsilon)^{k-1},
  \end{equation}
  where the probability is over the randomness of the honest parties.

To this end, suppose that
Inequality~\eqref{eqn:sizeH} holds.
Denote by $\event$ the event that
\begin{equation}\label{eqn:event}
\abs{\Hc_{i,j-1}}\geq\frac{2^{dnN L}}{2^{(k-1)NL}\cdot \left(\frac{4nN}{\epsilon}\right)^{k-1}\cdot 2^{2^m}}.
\end{equation}
By Inequality~\eqref{eqn:sizeH},
$$
\Pr[\event]\geq(1-\epsilon)^{k-1}.
$$
Therefore,
\begin{align*}
&\SD(R',\U_L)\leq\\
&\SD((R'\mid\event),\U_L)\cdot\Pr[\event]+\SD((R'\mid\neg{\event}),\U_L)\cdot\Pr[\neg{\event}]\leq \\
&\SD((R'\mid\event),\U_L)+\Pr[\neg{\event}]\leq\\
&\SD((R'\mid\event),\U_L)+1-(1-\epsilon)^{k-1}.
\end{align*}
This, together with the definition of $\mu^*$ (see Equation~\eqref{eqn:mu*}), implies that in order to prove Equation~\eqref{SD:R'} it suffices to prove that
$$
\SD((R'\mid\event),\U_L)\leq\sqrt{\epsilon}.
$$
This, together with Claim~\ref{claim:ent2sd}, implies that it suffices to prove
that
\begin{equation}\label{eqn:entR}
\ent(\lo'\mid\event)\geq L-\epsilon.
\end{equation}
To this end, let $H\leftarrow\Hc_{i,j-1}$.  Then,
\begin{align*}
&\ent(H\mid \event)\geq \\
&dnNL-(k-1)NL-(k-1)(\log {4nN})-(k-1)\log \frac{1}{\epsilon}-2^m=\\
&(dn-k+1)NL-(k-1)\left(\log {4nN}+\log \frac{1}{\epsilon}\right)-2^m,
\end{align*}
where the first inequality follows from Equation~\eqref{eqn:event} together with the definition of entropy (see Definition~\ref{def:entropy}), and the latter equality follows from basic arithmetics.

For every $\alpha\in[d]$ and every $\beta\in[n]$, we denote by $\row_{\alpha,\beta}\in\{0,1\}^{NL}$ the random variable obtained by choosing a random matrix  $H\leftarrow\Hc_{i,j-1}$, and setting $\row_{\alpha,\beta}$ to be the $\ith{(\alpha,\beta)}$ row of~$H$.
Note that
$$
\ent(H\mid \event)\leq\sum_{\alpha\in[d],\beta\in[n]}\ent(\row_{\alpha,\beta}\mid \event)\leq \ent(\row_{i,u}\mid \event)+NL(dn-k),
$$
where the first inequality follows from the basic property of Shannon entropy,  that for any random variables $X$ and $Y$, it holds that $\ent(X,Y)\leq\ent(X)+\ent(Y)$, and the second equality follows from the fact that $k-1$ of the rows in $\Hc_{i,j-1}$ are fixed.
This, together with the equations above, implies that
\begin{align*}
&\ent(\row_{i,u}\mid \event)\geq\\
& (dn-k+1)NL-(k-1)\left(\log {4nN}+\log \frac{1}{\epsilon}\right)-2^m-NL(dn-k)=\\
& NL-(k-1)\left(\log {4nN}+\log \frac{1}{\epsilon}\right)-2^m=\\
& NL-(k-1)\left(\log {4nN}+\log^2(dn)\right)-2^m.\\
\end{align*}

Recall that $(R'\mid\event)$ is the random variable defined by choosing $H\leftarrow\Hc_{i,j-1}$ (where we assume that event $\event$ holds for $\Hc_{i,j-1}$), choosing a random $\alpha\leftarrow[N]$, and setting $R'=H(i,u,\alpha)$.
Thus,
\begin{align*}
&\ent(R'\mid\event)\geq \\
&\frac{NL-(k-1)(\log {4nN}+\log^2(dn))-2^m}{N}=\\
&L-\frac{(k-1)(\log {4nN}+\log^2(dn))+2^m}{N}\geq \\
&L-\epsilon,
\end{align*}
proving \Cref{eqn:entR}, where the latter inequality follows from the definition of~$N$ (see Equation~\eqref{eqn:N}).\\

It remains to prove Inequality~\eqref{eqn:sizeH}.  We prove that Inequality~\eqref{eqn:sizeH} holds for any $(i,j)\in[d]\times\{0,1,\ldots,n\}$.  The proof is  by induction on $k=(i-1)\cdot n+j$.
The base case is $k=0$, which corresponds to $(i,j)=(1,0)$.
In this case, it is always holds that
$$
\abs{\Hc_{i,j}}=\abs{\Hc_{1,0}}=\abs{\Hc_1}\geq \frac{2^{dnN L}}{2^{2^m}},
$$
where the latter inequality follows from the definition of~$\Hc_1$ (see \Cref{eqn:H1}).

Next, assume that Inequality~\eqref{eqn:sizeH} holds for $k-1$, and we prove that it holds for~$k$.
Fix $i\in[d]$ and $j\in[n]$ such that $k=(i-1)\cdot n+j$.
By the induction hypothesis,
 \begin{equation*}
    \Pr\left[\abs{\Hc_{i,j-1}}\geq\frac{2^{dnN L}}{2^{(k-1)NL}\cdot \left(\frac{4nN}{\epsilon}\right)^{k-1}\cdot 2^{2^m}}\right]\geq (1-\epsilon)^{k-1}.
  \end{equation*}
We denote by $\event$ the event that indeed
  \begin{align*}
    \abs{\Hc_{i,j-1}}\geq\frac{2^{dnN L}}{2^{(k-1)NL}\cdot \left(\frac{4nN}{\epsilon}\right)^{k-1}\cdot 2^{2^m}}.
  \end{align*}
  Thus, by our induction hypothesis,
  \begin{align*}
    \Pr[\event]\geq (1-\epsilon)^{k-1}.
  \end{align*}

In what follows, fix {\em any} $\Hc_{i,j-1}$ such that event $\event$ holds.
Claim~\ref{claim:prob} (with  $U=\Hc_{i,j-1}$ and $M=2^{NL}\cdot 4nN$) implies
that
$$
\Pr\left[\abs{\Hc_{i,j}}\geq \frac{\abs{\Hc_{i,j-1}}}{2^{NL}\cdot{\frac{4nN}{\epsilon}}}\right]\geq 1-\epsilon.
$$
This, in turn, implies that
\begin{align*}
 &\Pr\left[\abs{\Hc_{i,j}}\geq\frac{2^{dnN L}}{2^{kNL}\cdot \left(\frac{4nN}{\epsilon}\right)^{k}\cdot 2^{2^m}}\right]\geq\\
 &\Pr\left[\abs{\Hc_{i,j}}\geq\frac{2^{dnN L}}{2^{kNL}\cdot \left(\frac{4nN}{\epsilon}\right)^{k}\cdot 2^{2^m}}\mid \event\right]\cdot\Pr[\event]\geq\\
 &\Pr\left[\abs{\Hc_{i,j}}\geq\frac{\abs{\Hc_{i,j-1}}}{2^{NL}\cdot{\frac{4nN}{\epsilon}}}\mid \event\right]\cdot\Pr[\event]\geq\\
 &(1-\epsilon)\cdot(1-\epsilon)^{k-1}=\\
 &(1-\epsilon)^k,
\end{align*}
as desired.
\end{proof}

\end{proof}

\subsection{Static Adversaries}\label{sec:static}
We note that Theorem~\ref{thm:main} holds also for static adversary. For completeness, we restate the theorem for static adversaries.

\begin{theorem}
 Fix any $m=m(n)$, $d=d(n)$, $L=L(n)$, and any $n$-party $d$-round public-coin
  protocol $\Pi$ that outputs messages in $\{0,1\}^m$ and in which all parties send
  messages of length $L=L(n)$.  Then, for any constant $\delta>0$, any $t=t(n)<n$, and any $s=s(m)$, if $\Pi$ is $(t,\delta,s)$-statically secure then there exists an $n$-party $d$-round
  $(t,\delta',s)$-statically secure public-coin protocol that simulates~$\Pi$, where all parties send messages of length $\ell =
  m\cdot\log^4(n\cdot d)$, and where $\delta'\leq\delta+\mu^*$ (where $\mu^*$ is the negligible function defined in Equation~\eqref{eqn:mu*}).
\end{theorem}

The proof is almost identical to the proof of Theorem~\ref{thm:main} except that in the static setting, the adversary $\Adv$ needs to decide which~$t$ parties to corrupt before the protocol begins.

Recall that in the proof of Theorem~\ref{thm:main}, the adversary~$\Adv$ simulates one of the adversaries $\Adv^H$.  In the static setting, the adversary $\Adv$ will choose to corrupt the $t$ parties that are consistent with as many $\Adv^H$ as possible.
More specifically, recall that in the proof of Theorem~\ref{thm:main} we defined $\Hc_1$ to be the set of all matrices $H$ such that $\Adv^H$ tries to bias the outcome towards a specific set~$M$. Recall that $\abs{\Hc_1}\geq\frac{\abs{\Hc}}{2^{2^m}}$.

In the static setting, for every $H\in\Hc_1$ we denote by $T^H$ the set of parties that the adversary $\Adv^H$  corrupts. For every set $T\subseteq[n]$ of size~$t$ let $$\alpha(T)=\abs{\{H\in\Hc_1: T^H=T\}}.
$$
We define
$$
T^*=\argmax_T\{\alpha(T)\},
$$
and the adversary $\Adv$ corrupts the set of parties $T^*$.
We define  $\Hc'_1\subseteq \Hc_1$ to consist of all the matrices $H\in\Hc_1$ for which $\Adv^H$ corrupts the set of parties~$T^*$.
Note that
$$
\abs{\Hc'_1}\geq\frac{\abs{\Hc_1}}{2^{n}}\geq\frac{\abs{\Hc}}{2^{2^m}\cdot2^{n}}.
$$

The rest of the proof is similar to that of Theorem~\ref{thm:main}, except that
the analysis is easier in the static setting, since the decision of who to
corrupt has already been made.

 \section{Public-Coin Protocols}\label{sec:public_coin}
 In this section we show how to convert any selection protocol into a
 public-coin protocol.

\begin{theorem}
  Every selection protocol $\Pi$ can be transformed into a protocol $\Pi'$ which
  simulates $\Pi$ and such that the messages sent in $\Pi'$ are uniformly
  random.  Moreover, the protocol $\Pi'$ preserves the security of $\Pi$ and its
  round complexity.
\end{theorem}

\begin{proofsketch}
  Let $\Pi$ be an $n$-party selection protocol. Let $d = d(n)$ be the number of
  communication rounds and let us assume for simplicity that each party speaks
  at each round. Assume, without loss of generality, that each party samples its
  own randomness ahead of time, when the protocol begins. That is, for every
  $j\in[n]$, party $\party{j}$ has randomness $r_j\in\B^{\ell}$, where we let
  $\ell$ be the maximum number of random bits used by all parties during the
  protocol. At each round~$i$, party $\party{j}$ evaluates a function $f_{i,j}$
  which depends on the transcript of the protocol so far, which we denote by
  $\trans_{i-1}$ (\ie $\trans_{i-1}$ are the messages sent by all parties in
  rounds $1,\dots,i-1$), and on its own randomness $r_j$. Namely, the message
  sent at round $i\in[d]$ by party $\party{j}$ is
  \begin{align*}
    m_{i,j} = f_{i,j}(\trans_{i-1}, r_j).
  \end{align*}

  Before we define the protocol $\Pi'$, we introduce some notation. We say that
  a random string $r$ is \emph{good} \wrt transcript $\trans_i$ and party
  $\party{j}$ if when it is used as the randomness of that party, it generates
  the same exact transcript.

  Next, we define the protocol $\Pi'$. In round $i\in[d]$, party $\party{j}$
  sends a uniformly random string $u_{i,j}$ of length $2^{\ell}\cdot \ell$. Specifically, each
  party sends a uniformly random permutation of all possible $\ell$-bit
  strings. At the end, after the $\ith{d}$ round ends, we interpret each
  $u_{i,j}$ as a collection of many possible random strings for party $\party{j}$,
  choose one (say the first), denoted by $r_{i,j}$, which is \emph{good} \wrt
  the transcript so far and think of the $\ith{(i,j)}$ message as
  $f_{i,j}(\trans_{i-1}, r_{i,j})$.

  First, we observe that the round complexity of $\Pi'$ is the same as that of
  $\Pi$. Next, we claim that in an honest execution (\ie in the absence of an
  adversary), the distribution of the output of the protocol $\Pi$ is identical
  to that of $\Pi'$ (namely, $\Pi'$ simulates $\Pi$). We first note that
  conditioned on the fact that a \emph{good} randomness was found for all
  $d\cdot n$ messages, the above distributions are the same. This is true since
  in $\Pi'$ each party sends all possible $\ell$ bit strings in a {\em uniformly
    random} order. Second, we note that, since each party sends \emph{all}
  possible $\ell$-bit strings in each round, there \emph{always} exists
  \emph{good} randomness.

  Next, we argue that the protocol $\Pi'$ is as secure as $\Pi$.
  This follows by a simple hybrid argument. We define a sequence of protocols
  $\Pi^{(i)}$ for $i\in\{0,\dots,dn\}$ in which until (and including) the
  $\ith{i}$ message, the parties act according to $\Pi$ and in the rest of the
  protocol they act according to $\Pi'$. Notice that $\Pi' \equiv \Pi^{(0)}$ and
  $\Pi \equiv \Pi^{(dn)}$. We argue that for every $i\in[dn]$, the ``advantage''
  of any $\Adv^{(i)}$ in $\Pi^{(i)}$ over any $\Adv^{(i-1)}$ in $\Pi^{(i-1)}$ is
  zero.

  To this end, observe that the first $i-1$ messages are distributed exactly the
  same.
  In the next message (\ie the $\ith{i}$ one) the protocols deviate. Assume
  party $\party{j}$ speaks in both. While in $\Pi^{(i)}$ the message sent is
  some function of the transcript so far and the initial randomness $\party{j}$
  has, in $\Pi^{(i-1)}$ it is a random permutation of all possible random
  strings. We first note that if party $\party{j}$ is corrupted, then both the
  adversary $\Adv^{(i)}$ and $\Adv^{(i-1)}$ can force any message in the name of
  $\party{j}$ and thus they have the same power in both protocols (recall that
  after the $\ith{i}$ message, the protocols are identical). Hence, assume that
  $\party{j}$ is not corrupted. In this case, the adversary $\Adv^{(i)}$ sees a
  message which is a function of the transcript up to that point and the
  (private) randomness of that party, whereas $\Adv^{(i-1)}$ sees a message
  which is a random permutation of all possible random strings. The theorem now
  follows by observing that one adversary can simulate the view of the other,
  and recalling that the rest of the messages in both protocols are
  identically distributed.
\end{proofsketch}

\subsubsection*{Acknowledgments}
We thank Nancy Lynch, Merav Parter and David Peleg for helpful remarks and
pointers. The second author thanks his advisor Moni Naor for his continuous
support.

\addcontentsline{toc}{section}{References}
\bibliographystyle{alpha}
\bibliography{RandDist}

\end{document}